\documentclass[12pt]{article}
\usepackage{amsmath}
\usepackage{graphicx}
\usepackage{enumerate}
\usepackage{natbib} 
\usepackage{url} 

\newcommand{\blind}{1}

\usepackage{epsfig}
\usepackage{epstopdf}
\usepackage{graphics,graphicx,color}
\usepackage{amssymb,amsfonts,amsthm,amsmath}
\usepackage{multirow}
\usepackage{mathrsfs}
\usepackage{indentfirst}
\usepackage{subfig}
\usepackage{float}
\usepackage{hyperref}
\usepackage{tabulary}
\usepackage[title]{appendix}
\usepackage{booktabs, siunitx,caption}
\usepackage[overload]{empheq}
\usepackage{cleveref}

\usepackage{algorithmicx,algpseudocode, algorithm}

\usepackage{xcolor}
\definecolor{DarkGreen}{rgb}{0.1,0.5,0.1}

\newcommand{\revision}[1]{\color{black}{#1}\color{black}}
\newcommand{\new}[1]{\color{black}{#1}\color{black}}

\newtheorem{theorem}{Theorem}

\newtheorem{lemma}{Lemma}

\begin{document}

\def\spacingset#1{\renewcommand{\baselinestretch}%
{#1}\small\normalsize} \spacingset{1}


\if1\blind
{
\title{\bf Bandit Change-Point Detection for Real-Time Monitoring High-Dimensional Data Under Sampling Control}

\author{Wanrong Zhang \hspace{.2cm}\\
Harvard University\\
	and \\
	Yajun Mei
	\hspace{.2cm}\\
Georgia Institute of Technology}
\date{}
  \maketitle
} \fi

\if0\blind
{
  \bigskip
  \bigskip
  \bigskip
  \begin{center}
    {\LARGE\bf Bandit Change-Point Detection for Real-Time Monitoring High-Dimensional Data Under Sampling Control}
\end{center}
  \medskip
} \fi

\bigskip
\begin{abstract}
In many real-world problems of real-time monitoring high-dimensional streaming data, one wants to detect an undesired event or change quickly once it occurs, but under the sampling control constraint in the sense that one might be able to only observe or use selected components data for decision-making per time step in the resource-constrained environments.
In this paper, we propose to incorporate multi-armed bandit approaches into sequential change-point detection to develop an efficient bandit change-point detection algorithm based on the limiting Bayesian approach to incorporate a prior knowledge of potential changes. Our proposed algorithm, termed Thompson-Sampling-Shiryaev-Roberts-Pollak (TSSRP), consists of two policies per time step: the adaptive sampling policy applies the Thompson Sampling algorithm to balance between exploration for acquiring long-term knowledge and exploitation for immediate reward gain, and the  statistical decision policy fuses the local Shiryaev-Roberts-Pollak statistics to determine whether to raise a global alarm by sum shrinkage techniques. Extensive numerical simulations and case studies demonstrate the statistical and computational efficiency of our proposed TSSRP algorithm.
\end{abstract}

\noindent%
\vfill

\newpage
\spacingset{2} 
\section{Introduction}

Real-time monitoring high-dimensional streaming data under sampling control constraints appears in many important applications such as
intrusion detection in computer networks \citep{bass1999multisensor}, event detection in social networks \citep{viswanath2014towards}, epidemic disease outbreak monitoring \citep{yang2015accurate}, anomaly detection in manufacture processes \citep{ding2006distributed}.
In these applications, one often can only observe or use selected components of the data for decision-making due to the capacity limitation in data acquisition, transmission, processing, or storage. For instance,  the sensor devices might have limited battery powers; thus, one might want to use a subset of sensors per time step over a long period instead of using full sensors simultaneously over a short period. Likewise, while sensing is usually cheap, the communication bandwidth is often limited from remote sensors to the fusion center that makes a global decision. The fusion center might prioritize certain local sensors to send local information for decision making.
Also, in many applications such as quality engineering or biosurveillance, one faces the design issue and needs to decide which variables or patients to be measured to detect the defect or disease outbreak more efficiently.

In this work, we investigate how to efficiently real-time monitor high-dimensional streaming data under resource constraints. We assume that the full data from a system is a $K$-dimensional random vector ${\bf X}_{t}=(X_{1,t}, \cdots, X_{K,t})$ at each time step,  but we can only observe $q$ out of $K$ components per time step. Here the component $X_{k,t}$, with $k=1,\ldots, K$ and $t=1,2,\ldots$, can be either the raw data from local sensors or the derived features such as wavelet coefficients, principal components. Initially, the system is in control in the sense that $X_{k,t}$ follows a probability density function $f_{k}.$ At some unknown time $\nu,$ an event may occur and change the distributions of a sparse subset of the components.  Our goal is to design an efficient algorithm to adaptively decide which variable to sample at each time step, and when to raise a global alarm to indicate the possible occurrence of the change.

Without resource constraints, monitoring fully observed streaming data has raised much attention in the statistical quality control (SPC) and sequential change-point detection literature, see \cite{zou2009multivariate, li2019two, li2020efficient, li2020diagnostic, zou2015efficient}. The existing work generally falls into two frameworks: the cumulative sum (CUSUM) type method, which is based on the generalized likelihood ratio (GLR) framework; and the Shiryaev-Roberts type method, which is based on the Bayesian framework. For classical research on one-dimensional data streams,
see \cite{shiryaev1963optimum, lorden1971procedures, pollak1985optimal, lai1995sequential, lai1998information, basseville1993detection,poor2008quickest, tartakovsky2014sequential}.
For recent research on high-dimensional data streams with fully observed data,  see \cite{zhang2012model, xie2013sequential, wang2015large, cho2015multiple, chan2017optimal,chu2019asymptotic}. Additionally, another framework is to monitor each data stream separately by computing respective local detection statistics and then fuse local statistics into a global-level monitoring statistic, see \cite{mei2010efficient, mei2011quickest,  liu2016scalable, li2020efficient}. This framework can balance the tradeoff between computational efficiency and statistical efficiency.

Real-time monitoring high-dimensional {\it partially} observed data streams under the sampling control has been studied in the literature of statistical process control in applied statistics. A prominent line of work is based on the CUSUM procedure for observed local streams with an artificially introduced compensation parameter for the unobserved local stream, see \cite{liu2015adaptive, xian2018nonparametric, wang2018spatial,xian2021online}. While the compensation parameter can increase the chance of exploring unobserved local streams, tuning the parameter is challenging. Another line of work leverages extra information such as the correlation structure to approximate unobserved local streams and then to plug in the standard monitoring methods of fully observed data, see \cite{Zhang:Hoi:2019, Nabhan:Mei:Shi:2021}.

We propose a bandit change-point detection algorithm for efficient real-time monitoring of high-dimensional streaming data under the sampling control. Our contributions are twofold: (i) We incorporate prior knowledge of potential changes to update unobserved local streams by treating the likelihood ratios of unobserved data as one.  (ii) We incorporate the Thompson Sampling algorithm in the multi-armed bandit (MAB) problem into the Shiryaev-Roberts-Pollak procedure in the sequential change-point detection literature.  While Bayesian methods often involve extensive computations, our method, termed as Thompson-Sampling-Shiryaev-Roberts-Pollak (TSSRP) algorithm, is computationally efficient when monitoring high-dimensional data when the data streams are mutually independent and we have some prior knowledge on the post-change distribution. The limiting Bayesian framework allows our algorithm to have a natural interpretation and avoid the tuning of artificial tuning parameters such as the compensation parameters for unobserved data. It can balance the tradeoff between exploiting the observed local components that maximize the immediate detection performance and exploring not-been-monitored local components that might provide new information to improve future detection performance. In particular, our proposed TSSRP algorithm performs similar to random sampling in the in-control state when no changes occurring, but becomes a greedy sampling on those affected local components in the out-of-control state when a change occurs. Numerical simulations and case studies show the efficiency of our proposed TSSRP algorithm.

The classical MAB problem focuses on developing algorithms to balance the tradeoff between exploration for acquiring long-term knowledge and exploitation for immediate reward gain,  see \cite{lai1985asymptotically, robbins1985some,scott2010modern, gittins2011multi, bubeck2012regret,agrawal2012analysis,cao2019nearly, zhao2019multi} and references therein. Our work and the classical MAB problem both deal with the dynamical/adaptive sampling strategy that samples those local streams with the largest values of some suitable local statistics. Nevertherless, our work is different from the existing research on the MAB with non-stationary or piecewise constant rewarding functions, see \cite{cao2019nearly, Ghatak:2020}, because our primary objective is to minimize the average detection delay subject to controlling false alarm rate, whereas the bandit problems minimize cumulative regret. In summary, we apply the bandit ideas to develop a new sequential change-point detection algorithm for monitoring partially observed data.

The remainder of this paper is organized as follows. In Section \ref{background}, we provide the mathematical formulation of our problem and also review the background of the multi-armed bandit problem and the sequential change-point detection problem. Next, we introduce our proposed method and develop its theoretical properties in Section \ref{bigmethod}. Then we evaluate the performance of our proposed algorithm through simulation studies and real data case studies in Section \ref{experiment} and \ref{realdata}, respectively. Concluding remarks are included in Section \ref{conclusion}. We provide the detailed proofs of all theorems in the Supplementary Materials, which also include additional numerical experiments.

\section{Problem Formulation and Backgrounds}\label{background}
In this section, we present the mathematical formulation of real-time monitoring high-dimensional streaming data in resource-constrained environments in Subsection \ref{formulation}. Then we provide a brief review of the Thompson Sampling algorithm for the multi-armed bandit problem in Subsection \ref{banditreview}, followed by the review of the Bayesian approach for sequential change-point detection in Subsection \ref{shiryaev}.

\subsection{Problem Formulation}\label{formulation}

Suppose we are monitoring $K$ independent data streams in a system. Let $X_{k,t}$ denote the observation of the $k$-th data stream at time $t$ for $t=1,2,\ldots$ and $k=1,2,\ldots,K$. Here each data stream $X_{k,t}$ can be the raw data itself or its derived features such as wavelet coefficients, principal components. Samples from data streams are assumed mutually independent. Each data stream generates identically distributed data from a specific distribution $f_{\theta_{k}}$. At some unknown change time $\nu \in \{1, 2, \ldots\}$, an undesirable event occurs and changes the distributions of some data streams abruptly in the sense of changing the values of the parameters $\theta_{k}$. Conditional on the change time $\nu >1$, for those affected data streams, the observation $X_{k,1}, \ldots, X_{k,\nu-1} $  are independent and identically distributed (iid) with density $f_{\theta_{k,0}}$ while $X_{k,\nu},  X_{k,\nu+1}, \ldots$ are iid with another density $f_{\theta_{k,1}}$, where $\theta_{k,1}> \theta_{k,0}$, and the case where $\theta_{k,1}< \theta_{k,0}$ can be handled similarly. For those unaffected data streams, all observations $X_{k,1}, X_{k,2} \ldots$ are iid with density $f_{\theta_{k,0}}$. Here we do not know which subset of the local streams changed, and we assume that the affected local streams are sparse. In practice, practitioners usually specify $f_{\theta_{k,1}}$ as the interested-smallest magnitude of a change to be detected.

Under sampling control, we can only observe or use selected partial data for decision making. We assume that only $q < K$ data streams can be selected to collect data at each time $t$. Mathematically, let $\delta_{k,t}$ be the indicator function that $\delta_{k,t}=1$ if and only if the $k$-th data stream $X_{k,t}$ is selected at time step $t$. The resource constraint implies that $\sum_{k=1}^{K}\delta_{k,t}=q$ at each time step $t=1,2,\cdots$. Imagine that we put $q$ sensors onto $K$ locations, then the $\delta_{k,t}$ can be thought as whether to put a sensor to the $k$-th data stream at time $t$. Let $S_t$ denote the locations where $\delta_{k,t}=1$, and we refer to it as the sensor layouts. Under our notation, the observed data can be represented as $\{X_{k,t}^\ast\}=\{X_{k,t}\delta_{k,t}\}$, $k=1,2,\ldots, K$.

For the change-point detection problem under sampling control, a statistical scheme consists of two policies per time step. One is the adaptive sampling policy $\delta$ that decides the observable location $\delta_{k,t}$, and the other is the statistical decision policy, often defined as a stopping time $T,$ that raises an alarm based on the observations available $\{X_{k,t}^\ast\}=\{X_{k,t}\delta_{k,t}\}_{1 \le k \le K, 1 \le t \le T}$. Our objective is to design a scalable and efficient statistical scheme of $(\delta, T)$ that minimizes the average detection delay
\begin{equation}\label{delay}
D(T)=\sup_{1\le\nu<\infty}\mathbb{E}_{\nu}(T-\nu|T \ge \nu),
\end{equation}
 subject to the Average Run length (ARL) to False Alarm constraint
 \begin{equation}\label{falsealarm}
 \mathbb{E}(T|\nu=\infty)\ge \gamma,
 \end{equation}
where $\gamma$ is a pre-specified constant.

It is worth noting that our problem connects to the multi-armed bandit problem in the sampling policy; however, they are fundamentally different due to different performance criteria.

\subsection{Thompson Sampling for Multi-Armed Bandit}\label{banditreview}

In this subsection, we briefly review the multi-armed bandit (MAB) problem, first introduced by \cite{robbins1952some}, and one of the most popular algorithms, Thompson Sampling \citep{thompson1933likelihood}. Under a classical setting of MAB, a gambler can play one of $K$ slot machines (or arms) for $K \ge  2$, but she or he has no prior knowledge about which machine has a potentially higher reward. The only way to learn rewards is to play the machines. The problem of interest is how the gambler decides which arm to play at each time step, to maximize the total rewards through $N$ plays.

One of the most popular MAB algorithms is Thompson Sampling, which is a natural Bayesian algorithm. It has been widely used for personalized advertisements and product recommendation \citep{agrawal2013thompson}, as its efficiency has been well demonstrated in many real-world applications, especially in the high-dimensional setting.  See \cite{scott2010modern}, \cite{chapelle2011empirical}. In particular,  \cite{agrawal2012analysis} showed that the Thompson Sampling algorithm asymptotically minimizes the expected regret.

The idea of Thompson Sampling is to sample arms based on the largest values of the random realizations of the posterior distributions instead of the posterior means. Specifically, at each time step, after updating the posterior distribution of the mean $\theta_{k}$ for each arm, we randomly sample a realization from the posterior distributions, denoted by $\hat{\theta}_{k}$ from the $k$-th arm. Then we select the arm with the largest random realization, i.e., $\arg \max_{1 \le k \le K} \hat{\theta}_{k}$. This allows us to have better chances to sample those arms with fewer observations, thereby balancing the tradeoff between the exploration for acquiring long-term knowledge and the exploitation for immediate reward gain.

In the multi-armed bandit problem, when we are allowed to observe $q$ arms each time, it is natural to extend the original Thompson Sampling algorithm to select the $q$-largest realizations. Such an approach often holds nice properties under reasonable conditions, see \cite{anantharam1987asymptotically, pandelis1999optimality, kaufmann2016complexity}. Thus we will adopt the Thompson sampling with $q$-largest realizations in our context.

\subsection{Shiryaev-Roberts Procedure}\label{shiryaev}

We now review the Bayesian approach for the simplest sequential change-point detection problem pioneered by \cite{shiryaev1963optimum}, as well as the corresponding limiting Bayesian approach. See \cite{roberts1966comparison,pollak1985optimal, pollak1987average}. Consider the simplest univariate case when we observe a sequence of independent observations $X_1, X_2, \ldots$, whose distribution might change from $f_{\theta_0}$ to $f_{\theta_1}$ at some unknown time $\nu$. Since the goal is to detect the change time $\nu$ quickly, the statistical procedure is defined as a stopping time $T$ with respect to the observed data $\{X_t\}_{t \ge 1},$  where $\{T=t\}$ means that we raise the alarm at time $t$ to indicate that a change has occurred up to time $t.$

Under the Bayesian formulation, it is assumed that the change-point $\nu$ has a geometric prior distribution:
\begin{eqnarray}\label{bayesian.review}
P(\nu = t) = p(1-p)^{t-1} \quad \mbox{ for } \quad t= 1,2,3,\cdots
\end{eqnarray}
where $0 < p < 1$ is a pre-specified constant. Moreover, conditional on the (unknown) change-point $\nu$, the pre-change observations, $X_1,\ldots,X_{\nu-1},$ are iid with density $f_{\theta_0}$ and are independent of the post-change observations, $X_\nu, X_{\nu+1},\ldots$ which are iid with density $f_{\theta_1}$. Assume that the cost of per post-change observation is $c > 0.$ Then the Bayesian formulation of sequential change-point detection is to find a statistical procedure $T$ that minimizes the Bayes risk $P(T<\nu)+c\mathbb{E}(T-\nu)^+$.

\cite{shiryaev1963optimum} first solved this problem, and the Bayesian solution is to raise an alarm at the first time when the posterior probability of change having occurred, i.e., $P(\nu\le t| X_1, \ldots, X_t)$, is greater than a certain threshold. Under the limiting Bayes framework, one considers the test statistic of the form $P(\nu\le t| X_1, \ldots, X_t)/ p(1- P(\nu\le t| X_1, \ldots, X_t))$  as $p$ goes to zero. This yields the so-called Shiryaev-Roberts procedure (\cite{roberts1966comparison}) that raises an alarm at time
\begin{eqnarray}
T_A=\inf\{t|R_t\ge A\},
\end{eqnarray}
where $R_{t}$ is the Shiryaev-Roberts statistic defined as
\begin{eqnarray}
R_t=\sum_{j=1}^{t}\prod_{i=j}^{t}\frac{f_{\theta_1}(X_i)}{f_{\theta_0}(X_i)},
\end{eqnarray}
and the threshold  $A$ is a pre-specified constant.
\cite{pollak1985optimal, pollak1987average} showed that this procedure enjoys nice asymptotic minimax properties, i.e., minimize the worst average detection delay in (\ref{delay}) up to within an $o(1)$ term subject to the ARL to false alarm constraint in (\ref{falsealarm}), as $\gamma$ goes to $\infty$.

\section{Bandit Change-Point Detection}\label{bigmethod}

In this section, we present our proposed TSSRP algorithm for the real-time monitoring high-dimensional streaming data under sampling control.
Our proposed algorithm can be thought of as the limit of Bayesian procedures that adapt the Thompson sampling policy of sampling local streams based on the random realizations of the posterior distributions.

We assume the data contains $K$ independent data streams. The local change time $\nu_k$ of the $k$-th local data stream has a prior Geometric($p$) distribution. The $k$-th local stream has an initial prior probability $\Pi_{k,0}$ of change and $\Pi_{k,0}$'s are mutually independent, and identically distributed from a common prior $G = G_{p}$. We can extend the idea to other non-homogeneous scenarios, e.g., in quality control of $K$-stages manufacturing process where some stages are more prone to defect. Specifically, the distribution of the local change time $\nu_{k}$ is as follows:
\begin{eqnarray}\label{bayesian.review}
P(\nu_k=0) = \Pi_{k,0},\quad  P(\nu_k = t) = (1- \Pi_{k,0}) p(1-p)^{t-1} \quad \mbox{ for } \quad t= 1,2,3,\cdots,
\end{eqnarray}
where $\Pi_{k,0}$ can be either a constant, e.g., $\Pi_{k,0} \equiv 0$, or a random variable that has a distribution $G_{p}.$

After taking observations at the time step $t,$ we update the posterior distribution of $\nu_{k},$ denoted by $\Pi_{k,t} = P(\nu_{k} \le t | \mbox{ Observed Data}).$  This computation is straightforward since the raw data $X_{k, t}$ is distributed as $f_{\theta_{k,0}} I( \nu_{k} < t) + f_{\theta_{k,1}} I(\nu_{k} \ge t),$ although it is observable if and only if the sampling indicator $\delta_{k,t} = 1$. Next, we combine the local posterior distributions $\Pi_{k,t}$'s together to decide if we would raise a global alarm. If yes, then we stop taking any observations. If no, then we will proceed to the next time step --- as in the Thompson sampling algorithm, we decide to observe those local streams in time step $t+1$ with the largest values of the posterior values $\Pi_{k,t}$'s if the initial values $\Pi_{k,0}$ are constant, or their random realizations if the initial values $\Pi_{k,0}$ are random variables.

Implementing the Bayesian algorithm in a naive way is computationally infeasible. Our algorithm overcomes this challenge by leveraging the property that the limit of the Bayesian algorithms as $p$ goes to $0$ has a mathematical equivalent representation that is computationally scalable. Therefore, our algorithm is both statistically and computationally efficient.

We present our proposed TSSRP methodology in Subsection \ref{method} and discuss the choice of parameters and prior distribution in Subsection \ref{parameters}. We develop the theoretical properties of our proposed TSSRP algorithm including its connection to the Bayesian procedures in Subsection \ref{property_dis}.

\subsection{Methodology Development}\label{method}

In the context of real-time monitoring high-dimensional streaming data under sampling control, a statistical procedure
consists of two policies per time step: (1) the adaptive sampling policy to decide the observation location; (2) the statistical decision policy to raise a global alarm based on the observed data. A common challenge in both components or policies is how to construct local statistics for each local stream that can guide us to make efficient decisions for both adaptive sampling and statistical decision policies. Our proposed method's key novelty is to recursively update two-dimensional local statistics over time that allow conveniently implement the Thompson Sampling.

\subsubsection{Local Statistics}\label{local}

We propose to recursively compute two-dimensional local statistics, denoted by  $R_{k,t}$ and $L_{k,t}$, at the $k$-th local data steam at each time step $t=1, 2, \cdots,$ where  $R_{k,t}$ mimics the classical Shiryaev-Robert statistics
\begin{eqnarray} \label{eq.R}
R_{k,t} = \left\{
  \begin{array}{ll}
    (R_{k,t-1}+1) \frac{f_{\theta_{k,1}}(X_{k,t})}{f_{\theta_{k,0}}(X_{k,t})}, & \hbox{if $\delta_{k,t}=1;$} \\
    R_{k,t-1}+1, & \hbox{if $\delta_{k,t}=0.$}
  \end{array}
\right.
\end{eqnarray}
with initial value $R_{k,t=0} =0,$ and the statistics $L_{k,t}$ mimics the likelihood ratio function
and
\begin{eqnarray} \label{eq.likelihood}
	L_{k,t} =
	\left\{
	\begin{array}{ll}
	L_{k,t-1} \frac{f_{\theta_{k,1}}(X_{k,t})}{f_{\theta_{k,0}}(X_{k,t})}, & \hbox{if $\delta_{k,t}=1$;} \\
	L_{k, t-1}, & \hbox{if $\delta_{k,t}=0.$}
  \end{array}
\right.
\end{eqnarray}
with initial value $L_{k,t=0} =1.$

At a high-level, the local statistic $R_{k,t}$ mainly provides the evidence how likely a local change has occurred, whereas the other local statistic $L_{k,t}$ is related to the number of samples taken at a given local data stream. If $\delta_{k,t}=1$, i.e., if one takes observations from that specific local data streams, then the update on $R_{k,t}$ and $L_{k,t}$ follows the classical  Shiryaev-Robert or likelihood  statistics, respectively. On the other hand, if $\delta_{k,t}=0$, i.e., if we do not take local observations, then we recursively update $R_{k,t}$ and $L_{k,t}$ by adding or multiplying the constant $1$, respectively. The intuition is to treat $\frac{f_{\theta_{k,1}}(X_{k,t})}{f_{\theta_{k,0}}(X_{k,t})}$ as $1$ if $X_{k,t}$ is missing.

Note that the definition or computation of $(R_{k,t}, L_{k,t})$ depends on which local sensors will be observed, i.e., the values of sampling indicator variables $\delta_{k,t}$'s, which will be defined in the next subsection.

\subsubsection{Adaptive Sampling Policy}\label{allocation}

Our proposed adaptive sampling policy is as follows.

At each time step $t=0,1,2\cdots,$ we compute the two-dimensional local statistics, $(R_{k,t}, L_{k,t}),$  based on the observed data streams, and we also sample a randomized value $\tilde{R}_{k,t}$ from a pre-specified  prior distribution $G$. The $\tilde{R}_{k,t}$ can be treated as the ``initial values." 

When $G$ is a point mass density of $0,$ $\tilde{R}_{k,t} \equiv 0$ for all $k=1, \cdots, K.$ Otherwise, $\tilde{R}_{k,t}$  can be different across sensors, which can be viewed as the prior knowledge of how likely a local data stream is likely affected by the change.  
Next,
we compute a real-valued local statistic that determines the sampling policies:
\begin{equation}\label{eq.random}
R^*_{k,t}=R_{k,t} + L_{k,t} \tilde{R}_{k,t}.
\end{equation}
Finally, at time step $t+1,$ we follow the Thompson Sampling to adaptively choose the local data streams with the largest $q$ values of $R^*_{(k),t}$ in (\ref{eq.random}).  Let $l_{(k),t+1}$ denote the corresponding index of the $k$th largest values, then the new sensor layout will be $S_{t+1}=\{l_{(1),t+1},\ldots,l_{(q),t+1}\}$ at time $t+1.$

Let us provide a high-level rationale for our proposed adaptive sampling policy. First, note that $R^*_{k,t}$ in (\ref{eq.random}) can be thought of as a randomized version of $R_{k,t}$ and allows us to balance better the tradeoff between those local streams having larger observed $R_{k,t}$ and those local streams having fewer observations. Second, our sampling policy is computationally efficient, since it is based on the recursive updates of the two-dimensional local statistics, $(R_{k,t}, L_{k,t}).$ Finally,  our sampling policy is the Thompson sampling method under the limiting Bayes framework, since the larger the $R^*_{k,t}$ value, the larger the realization of the posterior distribution of a local change.

\subsubsection{Global Decision}\label{stopping}

Our proposed global decision policy is to raise \revision{a global alarm } based on the largest $r$ values of the local statistics $R_{k,t}$ in (\ref{eq.R}), and is defined as the stopping time
\begin{equation}\label{stoppingtime}
T=\inf\{ t \ge 1: \sum_{k=1}^{r}R_{(k),t} \ge A\},
\end{equation}
where $r$ is a pre-specified parameter, and $A$ is a pre-specified constant so as to satisfy the false alarm constraint  in (\ref{falsealarm}). Here $R_{(1),t}\ge \ldots \ge R_{(k),t} \ge \ldots \ge R_{(K),t}$ denote the decreasing order of the local statistics  $R_{k,t}$ in (\ref{eq.R}).

We should acknowledge that there are many other ways to raise a global decision. For instance, for local statistics in the summation, we can use the randomized version $R_{(k),t}^*$ or the logarithm version $\log R_{(k),t}.$ Moreover, there are different ways to use the shrinkage transformation to combine local statistics to raise a  global alarm; see \cite{mei2011quickest} and \cite{liu2016scalable}. Based on our extensive simulation experiences, the stopping time in (\ref{stoppingtime}) is stable and outperforms other stopping rules in most cases. The discussion on comparison with other stopping rules is deferred to the supplementary material Section 3.

\subsubsection{Summary of Proposed Algorithm} \label{algsummary}

We summarize the proposed Thompson-Sampling-Shiryaev-Roberts-Pollak (TSSRP) in Algorithm \ref{TSSRPalg}.

\begin{algorithm}[t]
	\caption{Thompson-Sampling-Shiryaev-Roberts-Pollak (TSSRP) algorithm}
	\label{TSSRPalg}
	\textbf{Parameters}:  the number $r$, the number of observed sensors $q$, a prior distribution $G$ and the stopping threshold $A$.\\
	\textbf{Input}: $K$ data streams\\
    \textbf{Initialize:} Set $R_{k,t}=0$, $L_{k,t}=1$, and sample $\tilde{R}_{k,t}$ from $G$ for all $k=1,2,\ldots,K$. Randomly sample $q$ data streams as the initial layout $S_1$\\
	\textbf{Algorithm}:
	In each round $t \leftarrow 1,2,\ldots$ do the following:\\
	(1)  based on the current sensor layout $S_t$, recursively update two-dimensional local statistics  $(R_{k,t}, L_{k,t})$ in (\ref{eq.R}) and (\ref{eq.likelihood})\\
	(2) For each data stream $k$, sample the ``initial" value $\tilde{R}_{k,t}$ from $G$, and calculate the local sampling statistics $R^*_{k,t}$  in (\ref{eq.random})\\
	(3) Order the local sampling statistics  $R^*_{k,t}$ $k=1,2,\ldots,K$, from the largest to the smallest, and let $l_{(k),t}$ denote the variable index of the order statistics $R^*_{(k),t}$\\
	(4)  Update the sensor layout = $\{l_{(1),t},\ldots, l_{(q),t}\}$  \\
    \revision{ (5) Check if the criterion of the stopping time in (\ref{stoppingtime}) is reached. If yes, stop and raise a global alarm. If not, proceed to the next iteration.}
\end{algorithm}

\revision{Our proposed TSSRP method is not only statistical efficient as a limiting Bayesian procedure that is able to incorporate prior knowledge of potential changes, but also computationally scalable. First, it requires only $3K$ registers for retaining relevant information of $(R_{k,t}, L_{k,t}, \tilde{R}_{k,t})$ about the $K$ local processes: the first two are on the observed data regarding the local change, and the last is on the prior knowledge of the local change. Second, since the two-dimensional local statistics $(R_{k,t},L_{k,t})$ can be computed recursively and the ``initial" values  $\tilde{R}_{k,t}$ can be sampled directly from the prior distribution $G$ at each time step, the computational cost of our TSSRP method is linear to the number $K$ of local data streams. Thus our method can be easily implemented for real-time monitoring.
}

\subsection{Choice of Parameters}\label{parameters}

The TSSRP algorithm involves several parameters. Below we will discuss the choice of these parameters.

\textbf{Choice of the prior distribution}:  In practice, we could choose the prior distribution according to our prior knowledge. For example, in the manufacturing process, we may know that certain production lines could have a higher chance of being out of control. Alternatively, if no prior knowledge is present,  we could choose some non-informative priors such as the uniform distribution or the point mass $0$  distribution, i.e., $P(\tilde{R}_{k,t}=0)=1$. In the latter case, it reduces to a greedy sampling algorithm without randomization. 
In our numerical simulation studies in Section \ref{experiment}, we compare the performance of TSSRP with four different priors. The results suggest that the TSSRP significantly reduces the average detection delay regardless of the choice of the priors and that a valid prior can further improve the performance.

We note that \cite{pollak1985optimal} also investigates the choice of the prior distribution $G$, but under a different context in which the randomized  Shiryaev-Roberts statistics leads to an equalizer stopping time $T$ in that sense that $\mathbb{E}_{\nu}(T-\nu | T \ge \nu)$ is constant as a function of candidate change-point $\nu.$ Unfortunately, it is generally challenging to find an explicit solution for such prior distribution. Nevertheless, our purpose of the randomization is different from \cite{pollak1985optimal}: ours is for balancing the exploration and exploitation, while theirs is for almost minimax property.

\textbf{Choice of $r$}:
Intuitively, the tuning parameter $r$ in the stopping time (\ref{stoppingtime}) decides how many local sensors should be involved in the final decision making. Thus an ideal choice should be a plausible approximation of the actual number of changed data streams. If $r$ is much smaller than the actual number of changed data streams, our final decision will not use all information that the data might provide. If $r$ is much larger, our final decision will involve unnecessary noisy local statistics and lead to poor performance.
	Meanwhile, from the in-control performance viewpoint, as discussed in \cite{WuYanhong:2019}, the Shiryaev-Roberts statistic is heavy-tailed under the in-control hypothesis, and thus there is no practice difference between a smaller value of $r$ (e.g., $r=3$) to a larger value of $r$ (e.g., $r= q$). Thus, when the total number of changed data streams is unknown, one might simply choose $r=q.$

\textbf{Choice of $A$}: The parameter $A$ is the stopping threshold that controls the average run length to false alarm of our method, which is analogous to controlling the type I error. In practice, the threshold $A$ is usually determined by Monte-Carlo simulation. One often uses the bisection method to find the smallest $A$ so that the proposed method satisfies the false alarm constraint in (\ref{falsealarm}).

\subsection{Theoretical Properties of TSSRP}\label{property_dis}

In this subsection, we provide some theoretical properties of the TSSRP algorithm: Theorem \ref{prop1} establishes the close relationship between our algorithm and the Bayesian procedures.  Theorem \ref{ARL} investigates the in-control average run length properties, whereas Theorem \ref{prop_incontrol} and Theorem \ref{prop_outcontrol} provide a deep understanding of the sensor layouts.

\revision{First, the following theorem provides theoretical bases for the proposed adaptive sampling in  our TSSRP algorithm.

\begin{theorem}\label{prop1}
Assume that the change time $\nu_k$ for the $k$-th data stream has a prior Geometric($p$) distribution in (\ref{bayesian.review}), where the initial prior probability $\Pi_{k,0}$ has a prior $G_{p}$ distribution. Suppose that $G_{p} / p \rightarrow G$ in distribution, then
$R^*_{k,t}$ in (\ref{eq.random}) with its random component $\tilde{R}_{k,t} \sim G$ has the same distribution as
	\begin{equation}\label{eq.random1}
	\lim_{p \to 0} \frac{\Pi_{k,t}}{ p(1-\Pi_{k,t})},
	\end{equation}
where $\Pi_{k,t} = {\bf P}(\nu_k \le t | X_{k,1}^{*}, \cdots, X_{k,t}^*)$ is the posterior estimation how likely a local change occurs, and $X_{k,t}^{*} = X_{k,t} \delta_{k,t}$ denotes the observed data.
\end{theorem}

The detailed proof of Theorem \ref{prop1} is presented in the Appendix A.1 in online supplementary material. By this theorem, sampling based on the largest values of local statistics $R^*_{k,t}$ in (\ref{eq.random}) is mathematically equivalent to sampling based on the largest values of random realizations of the posterior distribution $\Pi_{k,t}$ as $p \to 0,$ since $\frac{u}{p(1-u)}$ is a monotonic increasing function of $u.$ Hence, the adaptive sampling policy in the TSSRP algorithm is the limit of the Thompson Sampling policy.
}

Second, we investigate the ARL to false alarm of the proposed TSSRP algorithm in the following theorem.

\begin{theorem}[Average Run Length to False Alarm]\label{ARL}
$\mathbb{E}_{\infty}(T)\ge A/K$. Moreover, $\mathbb{E}_{\infty} (T) = O(A) $, where $O(A)/A$ is bounded as $A\rightarrow \infty$.
\end{theorem}

Theorem \ref{ARL} provides us guidance to select conservative upper and lower bounds of $A$. Specifically, $K\cdot ARL$ can serve as the upper bound in the bisection search to speed up the threshold choosing procedure.  The detailed proof of Theorem \ref{ARL} is given in Appendix A.2 in online supplementary material. Unfortunately, it remains an open problem to derive the bounds on the average detection delays.

Next, we investigate the properties for sensor layouts. Theorem \ref{prop_incontrol} shows that the sensor layout will go through all the data streams eventually when the system is in control.

\begin{theorem}\label{prop_incontrol}
	Let $S_t$ be the sensor layout at time $t$. Then under $H_0: \nu=\infty$, there exsits a $ t'>t$ such that  $P(k \in S_{t'})>0$, for each $t > 0$ and each $ k , 1\le k\le K$.
\end{theorem}

Theorem \ref{prop_incontrol} implies that each variable has a chance to be explored, regardless of the sensor deployments in previous steps when no changes occur. In other words, the algorithm performs similar to random sampling in the in-control state.

Finally, Theorem \ref{prop_outcontrol} below implies that the sensor layout of the TSSRP algorithm will eventually converge to the changed data streams when the system is out of control.

\begin{theorem}\label{prop_outcontrol}
	Let $S_t$ be the sensor layout at time $t$. Then under $H_1: \nu<\infty$, we have $P(k\in S_t, \forall t>t_0|k\in S_{t_0})>0$ for all changed data stream $k$, and all $t_0>\nu$.
\end{theorem}

We have shown that the sensor layout will not stay on any unchanged data streams forever in Theorem \ref{prop_incontrol}. The sensors will eventually be redistributed to the affected data streams at a certain time. Theorem \ref{prop_outcontrol} states that once a sensor is deployed to the out-of-control data stream, then there is a nonzero probability that the sensor will stay on this data stream forever. This property ensures that the sensors will eventually keep monitoring the affected data streams. The proofs of Theorem \ref{prop_incontrol} and Theorem \ref{prop_outcontrol} are given in Appendix A.3 in online supplementary material.

In summary, our TSSRP algorithm admits a nice property that it does more exploration when the system is in control, while more exploitation when the system is out of control. Intuitively, such nice property comes from the structure of the local statistics. When we observe a new data, the recursive formula for local statistics involves the likelihood ratio $f_{\theta_{k,1}}(x_{k,t})/f_{\theta_{k,0}}(x_{k,t})$. Under the in-control state, the expectation of the increment of all local statistics at time $t$ is $1$, i.e., $\mathbb{E}[R^*_{k,t}]=t$, $\forall k, t$. Thus our adaptive sampling is similar to random sampling under the in-control state. Under the out-of-control state, the likelihood ratio under the post-change distribution will be more likely to be greater than one. Therefore, the affected data streams' local statistics will increase faster than those of unaffected data streams, which enables the sensor layout to converge to those affected local components.

\section{Simulation Experiments}\label{experiment}

In this section, we report the numerical performance of the proposed TSSRP algorithm and compare it with the existing algorithms. The general setting for our simulations is as follows. We consider monitoring $K=100$ independent data streams. We assume that $q=10$ out of $K=100$ data streams can be monitored at each time step. The nominal value for the ARL to a false alarm is fixed as $\mathbb{E}_\infty T =\gamma = 1000.$ We follow the classic approach (\citealp{xie2013sequential,xie2013change,mei2010efficient}) and report the average detection delay when the change occurs at time $\nu=1$. It is the most challenging setting because it is more difficult to detect when the change occurred at the very beginning.

We report our simulation results in two subsections, depending on different true generative models of the data.
Subsection \ref{gaussian} focuses on the statistical efficiency of the TSSRP algorithm when our prior knowledge on the candidate affected local streams and the Gaussian distribution of the data are valid, and Subsection \ref{tdist} considers the robustness of our algorithm under the mis-specified models. In our simulation studies below, all numerical results are based on $1000$ Monte Carlo replications.

For our proposed TSSRP algorithm, we consider four choices for the prior distribution $G$ on the initial  statistics $\tilde{R}_{k,t=0}$, which is the prior knowledge of how likely a local stream might be affected by the change:
\begin{description}
  \item[(i)] $G_{0}:$ Uniform $U[0.5,1]$ for the first ten local streams, and Uniform $U[0,0.5]$ for the remaining local streams;
  \item[(i)] $G_{1}:$ Uniform $U[0.5,1]$ for the first five local streams, and Uniform $U[0,0.5]$ for the remaining local streams;
  \item[(ii)] $G_{2}:$ Uniform $U[0,1]$ for all local streams;
  \item[(iii)] $G_{3}:$ the point mass $0$ for all local streams, i.e., $P_0=P(\tilde{R}_{k,t}=0)=1$.
\end{description}
We referred them to as $\text{TSSRP}(G_{0})$, $\text{TSSRP}(G_{1})$, $\text{TSSRP}(G_{2})$ and $\text{TSSRP}(G_{3})$, respectively.

We compare our TSSRP algorithms with the baseline Top-r Based Adaptive Sampling (TRAS) algorithm proposed by \cite{liu2015adaptive}. The TRAS algorithm first constructs a local CUSUM statistic for each observed variable. For the unobserved variables, the local statistics are updated by adding a compensation parameter $\Delta$. That is, under our notation, each local stream computes a local statistic
\begin{eqnarray} \label{eq.TRAS.cusum}
W_{k,t} = \left\{
  \begin{array}{ll}
    \max(W_{k,t-1}+ \log \frac{f_{\theta_{k,1}}(X_{k,t})}{f_{\theta_{k,0}}(X_{k,t})}, 0), & \hbox{if $\delta_{k,t}=1;$} \\
    w_{k,t-1}+ \Delta, & \hbox{if $\delta_{k,t}=0.$}
  \end{array}
\right.
\end{eqnarray}
where $\Delta$ is the so-called compensation parameter for unobserved data streams. Next, the TRAS algorithm combines the top-$r$ local statistics to determine whether to raise an alarm, i.e., the stopping time of the TRAS algorithm is given by
\begin{equation}\label{stoppingtime.cusum}
\tau(a)=\inf\{ t \ge 1: \sum_{k=1}^{r} W_{(k),t} \ge a\},
\end{equation}
where $W_{(k),t}$'s are the order statistics of $W_{k,t}$'s in (\ref{eq.TRAS.cusum}). Moreover, the TRAS algorithm adaptively deploys the sensors to the data streams with $q$ largest local statistics $W_{k,t}$'s in (\ref{eq.TRAS.cusum}) at the next time step $t+1.$  As reported in \cite{liu2015adaptive}, the average detection delay performance of the TRAS algorithm is sensitive to the choice of the compensation parameter $\Delta,$ and it remains an open problem to decide how to choose it suitably. In our simulations below, we present results obtained by three different choices of $\Delta= 0.03, 0.05, 0.1$.

\subsection{Statistical Efficiency}\label{gaussian}

In this subsection, we focus on the statistical efficiency of the TSSRP algorithms when the prior knowledge is valid. We consider the scenario where we monitor $K=100$ independent Gaussian data streams whose pre-change distributions are $N(0,1)$, and the first $r$ local streams change to the post-change distribution $N(1.5,1).$ We compare our TSSRP algorithm against other procedures with the four choices of initial distribution $G$ and the correct post-change mean $\mu_1=1.5$.
We vary the number of changed data streams ranging from $1$ to $10$. As mentioned in Section \ref{parameters}, the parameter $r$ in the global stopping time defined in \eqref{stoppingtime} ideally should be the number of changed data steams. The latter, however, is usually unknown in practice. We report the simulation results of different global monitoring schemes under a large $r=10$ in Table \ref{tbl:efficiency}. Additional simulation results under a small $r=3$ is deferred to Table 5 in supplementary material. The corresponding standard errors are also included in these tables to characterize the average detection delay distribution.

\begin{table}[th]
	\centering
	\caption{average detection delay under various number of changed data streams for the evaluation of the statistical efficiency experiments when the data is independent multivariate Gaussian distributed. All the experiments are conducted under $r=10$.}
	\label{tbl:efficiency}
	\small
	\begin{tabular*}{1\textwidth}{@{\extracolsep{\fill}}lccccccccc@{}}\toprule
		The number of changes&1&3&5&8&10\\
		\midrule
		TSSRP($G_{0}$)&12.15(0.23)& 7.67(0.07)&6.66(0.05)&6.05(0.04)&5.81(0.03)\\
		TSSRP($G_{1}$)&12.06(0.23)&7.59(0.07)&6.75(0.05)&6.57(0.04)&6.49(0.04)\\
		TSSRP($G_{2}$)&18.84(0.33)&11.93(0.14)&10.05(0.11)&8.67(0.08)&8.22(0.07)\\
		TSSRP($G_{3}$)&19.43(0.35)&11.79(0.14)&9.84(0.11)&8.74(0.08)&8.04(0.07)\\
		TRAS($\Delta=0.03$)&36.12(0.60)&21.10(0.25)&17.01(0.20)&13.43(0.15)&11.87(0.13)\\
		TRAS($\Delta=0.05$)&36.79(0.54)&22.84(0.24)&18.52(0.18)&15.17(0.13)&13.52(0.12)\\
		TRAS($\Delta=0.1$)&63.43(0.44)&37.87(0.25)&30.47(0.18)&25.39(0.13)&22.89(0.12)\\
		RSADA&71.87(1.63)&36.61(0.69)&26.94(0.51)&21.20(0.38)&18.54(0.34)\\
		\bottomrule
	\end{tabular*}
\end{table}

We make two key observations from Table \ref{tbl:efficiency}. First, all four variants of our proposed TSSRP algorithms are statistically efficient in the sense of having significantly smaller average detection delays compared to other procedures. When incorporating the correct prior information, i.e., with the prior   $G_0$ and $G_1$ on the initial statistics $\tilde{R}_{k,t=0}$, the TSSRP algorithm achieves the smallest average detection delays, since it appropriately incorporates the Bayesian information on the spatial locations of changes. Even if we use non-informative priors such as $G_2$ or $G_3$, the TSSRP algorithm still provides good performance as compared to the baseline TRAS algorithm, suggesting that the choice of priors in the TSSRP algorithm can be flexible.  Second, fewer tuning efforts are required for the TSSRP algorithm, because we set the likelihood of the unobserved or missing data to be $1$ under the Bayesian framework. Moreover, the performance of the TSSRP algorithm is relatively stable to the tuning parameter $r$, e.g., the number of local sensors involved in the final decision making, see table 5 in the supplementary material. This is consistent with our intuition that the Shiryaev-Roberts statistics is the exponent of the CUSUM statistics, and the sum of top-$r$ Shiryaev-Roberts statistics mainly captures the maximum local statistics across all the data streams, since $\sum_{i=1}^{r} \exp(a_i) \sim \exp(a_1)$ for large values of ordered sequences $a_1 > a_2 > \cdots > a_r$. Therefore, our TSSRP algorithm is not only statistically efficient but also easy to tune and use in practice.

\subsection{Robustness}\label{tdist}

In this subsection, we focus on the robustness of the TSSRP algorithm when the underlying generative model is mis-specified. We focus on two cases: (i) TSSRP algorithm is constructed for the post-change mean lower bound $\mu_{1} = 1.5$ when the true post-change mean of Gaussian distribution is $\mu_{1, true} = 2$ and (ii) TSSRP algorithm is constructed for Gaussian distributions when the data follows $t$ distributions with degree of freedom $df=5$.

\begin{table}[t]
	\centering
	\caption{average detection delay under various number of changed data streams for robustness experiments when the post-change parameter is mis-specified. The data streams follows Gaussian distribution.}
	\label{table100large}
	\small
	\begin{tabular*}{1\textwidth}{@{\extracolsep{\fill}}lccccccccc@{}}\toprule
		The number of changes&1&3&5&8&10\\
		\midrule
		TSSRP($G_0$)&7.37(0.10)&5.43(0.03)&4.98(0.03)&4.54(0.02)&4.43(0.02)\\
		TSSRP($G_1$)&7.33(0.10)&5.33(0.03)&4.87(0.03)&4.77(0.03)&4.72(0.02)\\
		TSSRP($G_2$)&8.64(0.17)&5.84(0.07)&5.64(0.06)&5.49(0.05)&5.32(0.04)\\
		TSSRP($G_3$)&12.77(0.18)&8.28(0.09)&7.18(0.07)&6.16(0.05)&5.87(0.05)\\
		TRAS($\Delta=0.03$)&27.03(0.42)&16.42(0.21)&12.69(0.15)&10.03(0.11)&8.87(0.10)\\
		TRAS($\Delta=0.05$)&27.79(0.34)&17.42(0.18)&14.38(0.15)&11.10(0.11)&9.91(0.09)\\
		TRAS($\Delta=0.1$)&44.93(0.28)&27.73(0.17)&22.40(0.13)&18.78(0.11)&17.11(0.10)\\
		\bottomrule
	\end{tabular*}
\end{table}

The results for the first case are shown in Table \ref{table100large}, which summarizes the average average detection delay and the corresponding standard errors for TSSRP under four choices of priors and TRAS under $\Delta=0.03, 0.05,0.1$. As in the case in subsection \ref{gaussian}, we see that TSSRP outperforms TRAS and that performance of TSSRP generally improves as the prior information gets correct.  Moreover, compared with the results in Table \ref{tbl:efficiency}, we see that a larger magnitude of change is easier to detect, suggesting that we can use the smallest magnitude of change to specify the post-change parameters if they are unknown.

For the latter case, we determine the stopping threshold $A$ by 1000 Monte Carlo simulations under the $t$ distribution with mean $0$, and we choose $r=10$ to construct the global statistics.
Figure \ref{table100robust} plots the average detection delays against the number of changed data streams for TSSRP and TRAS, with varying priors ($G_1$ and $G_2$) and tuning parameters ($\Delta=0.03, 0.05$).  We observe a similar pattern as before, again reinforcing the usefulness of incorporating Bayesian information.

\begin{figure}[t]
	\centering
	\includegraphics[width=3.5in]{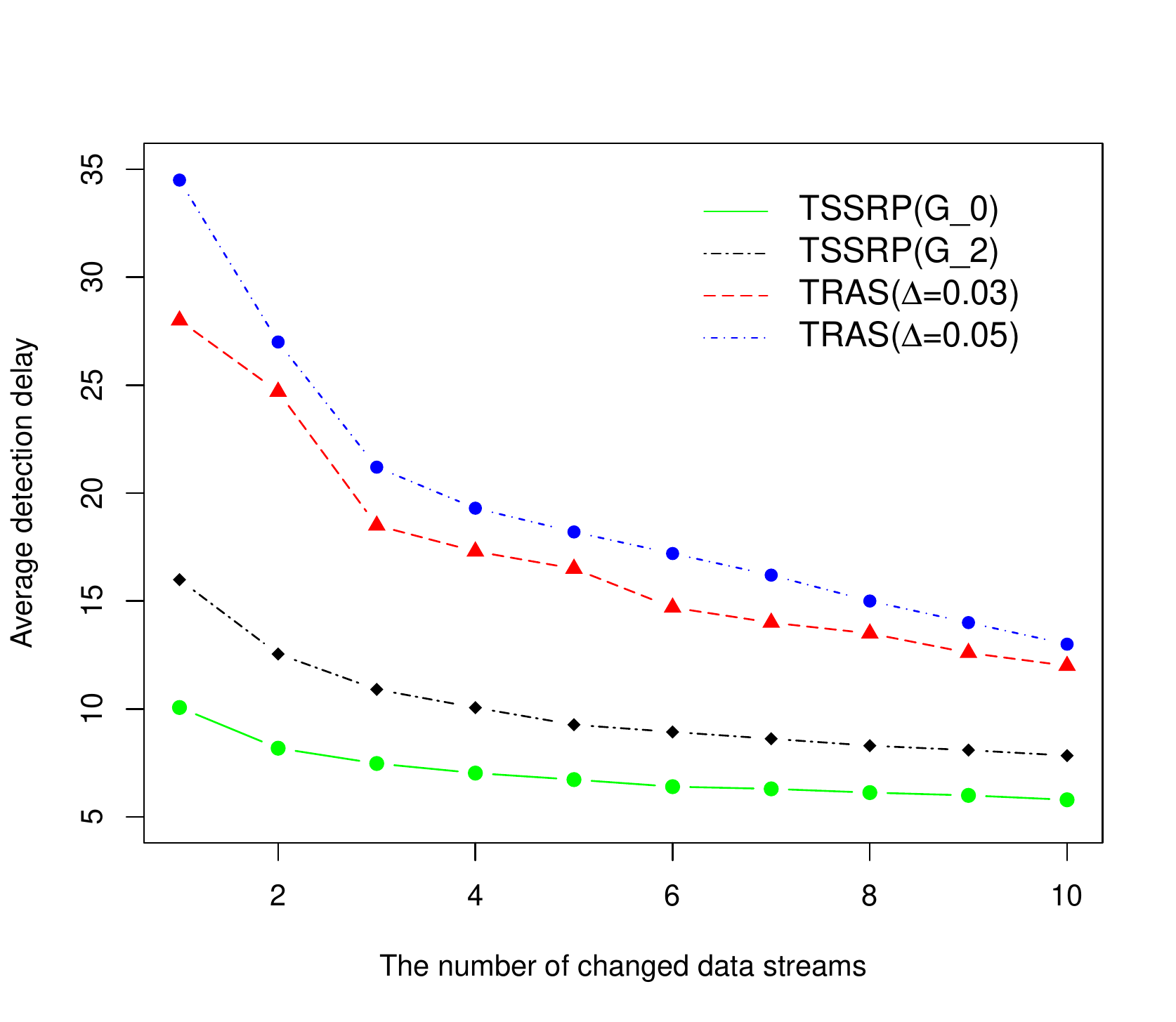}
	\caption{average detection delay versus the number of changed data streams for robustness experiments when the distribution is mis-specified. The data streams follows the $t$ distribution.}
	\label{table100robust}
\end{figure}

Our results suggest that our TSSRP algorithm can be robust when the distributional assumption is somewhat violated, and our TSSRP still outperforms the TRAS algorithm. This experiment further bolsters the stability of our algorithm, indicating that the hypothesized distributions can be slightly different from the underlying distributions of the data, and our algorithm will still raise an alarm quickly. In theory, the robustness of the TSSRP algorithm depends heavily on the robustness of the likelihood $\frac{f_1}{f_0}$. As a result, when the underlying models are significantly mis-specified, our TSSRP algorithm can be less efficient as compared to other robust methods such as the RSADA algorithm in \cite{xian2021online} that is based on ranks. In such a case, our algorithm can be extended to a more robust variant by considering the likelihood of sequential ranks as in \cite{gordon1995robust}, which is beyond the scope of this article, and we will leave it as future work.


\section{Case Study}\label{realdata}

In this section, we evaluate the performance of the TSSRP algorithm on a real case example: the hot forming process. We also give an additional solar flare detection example on the high-dimensional case in Appendix B in supplementary material.

\begin{figure}[t]
	\centering
	\includegraphics[width=2.5in]{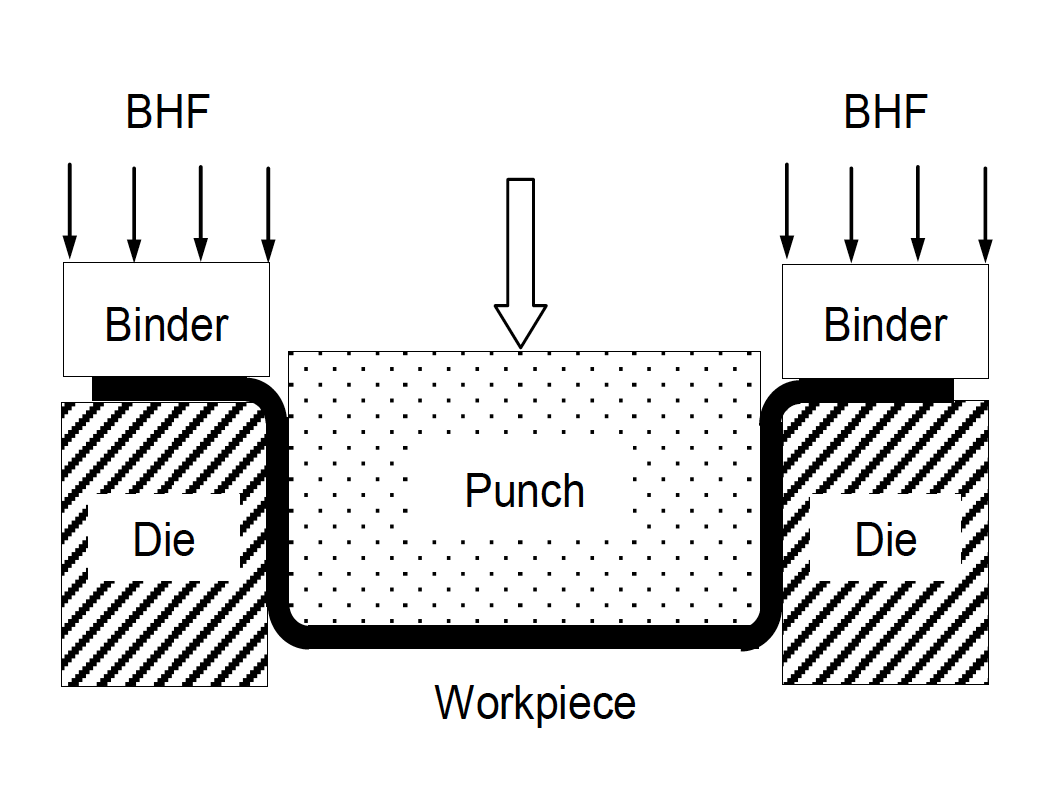}
	\includegraphics[width=2.5in]{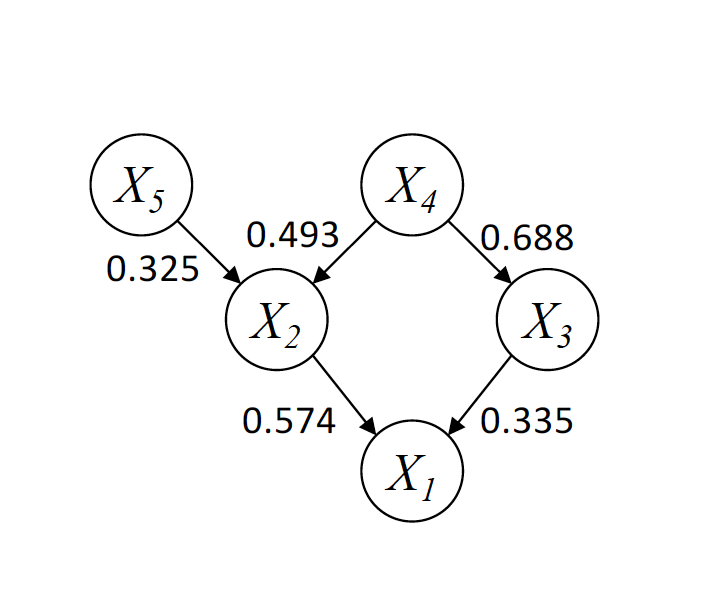}
	\caption{Left: 2-D illustration of the hot forming process. Right: Bayesian network for the hot forming process }
	\label{hotforming}
\end{figure}

We consider the Hot Forming Process example in \cite{li2010optimal}. We want to detect anomalies in this physical system. Figure \ref{hotforming} (Left) illustrates a two-dimensional (2-D) physical illustration of the hot forming process.  \cite{li2010optimal} identified the causal relationship of the five variables in this process: the final dimension of workpiece $X_1$, the tension in workpiece $X_2$, the material flow stress $X_3$, temperature $X_4$ and Blank Holding Force $X_5$, which can be represented as a Bayesian network. All the variables are proven to follow standard normal distribution when the system is under normal operating conditions. \cite{li2010optimal} gave the parameterization model of the Bayesian network, and Figure \ref{hotforming} Right illustrates the dependence across the five variables:
\begin{eqnarray}
X_i=\sum_{X_j\in p(X_i)} w_{j,i}X_j+\epsilon_i
\end{eqnarray}
where $p(X)=\{Y: Y\rightarrow X \in \text{Edge Set}\}$ denotes the parents of $X$; and $w(X_j, X_i)$ is the weight of the edge $X_j\rightarrow X_i$, which refers to the causal influence from $X_j$ to $X_i$; $\epsilon_i\sim N(0,\sigma_i^2)$ is the independent Gaussian noise.

In this study, we assume that the Bayesian network is unknown to us. We will only use the network structure to generate data under different scenarios.  We generate the changed root variables by setting the true mean change as two,  and the remaining variables according to the Bayesian network. In the algorithm, we set the post-change mean as the interested smallest shift magnitude  $\mu_{k,1}=1.5$ according to the characteristics of the actual system. We set the in-control ARL to be 100 and $r=2$ as in \cite{liu2015adaptive}. In each replicate, the changed data streams and the initial sensor layouts are selected randomly. We evaluate the average detection delay $D(T)$ as the average average detection delay under any change possibilities.

Table \ref{hotsim} summarizes the average detection delay comparisons between the TSSRP algorithm and the TRAS algorithm in the single change and two changes cases. It implies that the performance of our TSSRP algorithm is better than that of the baseline TRAS algorithm in this hot forming procedure.

\begin{table}[t]
	\centering
	\caption{Comparisons of the average detection delay between the TSSRP algorithm and the TRAS algorithm under different numbers of changed root variables in the hot forming process example}
	\small
	\label{hotsim}
	\begin{tabular*}{0.9\textwidth}{@{\extracolsep{\fill}}lccccccccc@{}}\toprule
		The number of changes&1&2\\
		\midrule
		TSSRP($U[0,1]$)&5.84(0.10)&4.54(0.06)\\
		TSSRP($P_0$)&6.48(0.09)&5.29(0.07)\\
		TRAS($\Delta=0.01$)&8.41(0.13)&7.39(0.11)\\
		TRAS($\Delta=0.1$)&8.17(0.13)&7.09(0.12)\\
		TRAS($\Delta=0.5$)&9.94(0.15)&8.78(0.14)\\
		\bottomrule
	\end{tabular*}
\end{table}
We are interested in studying how the sensor layouts update over the time under the in-control and out-of-control states empirically to validate Theorem \ref{prop_incontrol} and Theorem \ref{prop_outcontrol}. We summarize the average percentages of each data stream being observed under the two states in Table \ref{hotsim.sensor}, based on 1000 replicates. Specifically, it is defined as $$\text{percentage of being observed}=\frac{\# \text{time steps being observed}}{\# \text{total time steps until raising an alarm}}.$$ Here, the out-of-control state is when $X_1$ and $X_2$ change at the very beginning. The simulations further confirm that our TSSRP algorithm works similar to random sampling when the system is in control, and greedily selects the changed data streams when the system is out-of-control.

\begin{table}[t]
	\centering
	\caption{Sensor layouts distribution under the in-control and out-of-control states in the hot forming process example}
	\small
	\label{hotsim.sensor}
	\begin{tabular*}{0.9\textwidth}{@{\extracolsep{\fill}}lccccccccc@{}}\toprule
		Variable&Percentage (In-control)&Percentage (Out-of-control)\\
		\midrule
		$X_1$&0.408&0.718\\
		$X_2$&0.396&0.606\\
		$X_3$&0.400&0.226\\
		$X_4$&0.400&0.224\\
		$X_5$&0.396&0.226\\
		\bottomrule
	\end{tabular*}
\end{table}

\section{Conclusions and Discussion}\label{conclusion}

Processing high-velocity streams of high-dimensional data in resource-constrained environments is a big challenge. In this paper, we propose a bandit change-point detection approach to adaptively sample useful local components and determine a global stopping time for real-time monitoring of high-dimensional streaming data. Our proposed algorithm, termed Thompson-Sampling Shiryaev-Roberts-Pollak (TSSRP) algorithm, can balance between exploiting those observed local components that maximize the immediate detection performance and exploring not-been-monitored local components that might accumulate new information to improve future detection performance. Our numerical simulations and case studies show that the TSSRP algorithm can significantly reduce the average detection delay compared to the existing methods.

This work can be extended in several directions. First, based on the numerical simulation studies, we conjecture that our proposed TSSRP algorithm is first-order asymptotically optimal under a general setting. Still, it remains an open problem to prove it, as it is highly non-trivial to analyze the expected average detection delay of the proposed method.  Second, instead of a fixed number of active sensors, one could consider changing the number of active sensors per time step, and increase the number of sensors if a change likely occurs. Third, it is also interesting to find an optimal value of the number of active sensors that can adaptively adjust to make the best use of the resource.

Finally, we remark that our TSSRP algorithm is a computationally scalable representation of the limit of Bayesian procedures under the simplest model assumption where the data is i.i.d. and the post-change parameters are known. It can be extended to handle the case when the post-change parameters are unknown if we introduce a prior distribution on the post-change parameters. Furthermore, for more complicated models where the data streams have a spatial or temporal correlation structure, the proposed Bayesian and Thompson sampling framework can still be applicable if we can update the posterior distribution efficiently, say, via Markov chain Monte Carlo (MCMC). Therefore, our work opens a new research direction on statistical process control and sequential change-point detection when monitoring high-dimensional data streams under the sampling control.

\section*{Supplementary Materials}
In the supplementary materials, we provide (A) the detailed proofs of all theorems, (B) an additional case study in Solar Flare data, (C) additional simulation studies on our proposed algorithm when raising a global alarm based on the sum of the largest $r=3$ local statistics, and (D) the comparison with more global decision policies. The zip file contains R codes for our algorithm.

\section*{Acknowledgments}
The authors thank the editor, the associate editor, and two reviewers for their invaluable comments that greatly help to improve the article.

\section*{Funding}
W.Z. is supported in part by an ARC-TRIAD fellowship from the Georgia Institute of Technology, and a Computing Innovation Fellowship from the Computing Research Association (CRA) and the Computing Community Consortium (CCC). This work was completed while W.Z. was at Georgia Institute of Technology. Y.M. is supported in part by NSF grant  DMS-2015405. W.Z. and Y.M. are also supported in part by the National Center for Advancing Translational Sciences of the National Institutes of Health under Award Number UL1TR002378. The content is solely the responsibility of the authors and does not necessarily represent the official views of the National Institutes of Health.

\bibliographystyle{apalike}

\bibliography{monitoring}

\newpage

\begin{appendices}

\begin{center}
	
	\title{\large \bf Bandit Change-Point Detection for Real-Time Monitoring High-Dimensional Data Under Sampling Control}

	\author{Wanrong Zhang \hspace{.2cm}\\
		Harvard University\\
		and \\
		Yajun Mei\hspace{.2cm}\\
		Georgia Institute of Technology}
	
		{\large\bf SUPPLEMENTARY MATERIAL}
\end{center}

\new{In the supplementary materials, we provide missing proofs in \ref{app.proof}, an additional case study in \ref{solar}, additional simulation studies in \ref{app.sim}, and the comparison with more global decision policies in \ref{stoppingcomp}.}

\section{Missing Proofs}\label{app.proof}

\subsection{Proof for Theorem \ref{prop1}}\label{rule.app}

\begin{proof}
	Let us first investigate how to update the posterior probability $\Pi_{k,t}$. There are two subcases, depending on whether we take an observation from the $k$-th data stream or not. In the first case, when we do not take observations from the $k$-th stream, i.e., when $\delta_{k,t}=0$, we have the following recursive formula for the posterior probability $\Pi_{k,t}$.
	\begin{align*}
	\Pi_{k,t}&=P(\nu_k \le t | X^*_{k,1}, \cdots, X^*_{k,t-1})\\
	&=P(\nu_k \le t-1 | X^*_{k,1}, \cdots, X^*_{k,t-1})+P(\nu_k = t|\nu_k> t-1)P(\nu_k>t-1|X^*_{k,1}, \cdots, X^*_{k,t-1})\\
	&=\Pi_{k,t-1}+(1-\Pi_{k,t-1})p.
	\end{align*}
	The second case is when we take observations from the $k$-th stream, i.e., when $\delta_{k,t}=1$, by Bayes rule, we have the following the recursive formula:
	\begin{align*}
	\Pi_{k,t}=&P(\nu_k \le t|X^*_{k,1}, \cdots, X^*_{k,t-1}, X_{k,t})\\
	=&\frac{P(\nu_k \le t, X^*_{k,1}, \cdots, X^*_{k,t-1}, X_{k,t})}{P(X^*_{k,1}, \cdots, X^*_{k,t-1}, X_{k,t})}\\
	=&\frac{P(\nu_k \le t-1,  X^*_{k,1}, \cdots, X^*_{k,t-1}, X_{k,t})+P(\nu_k = t, X^*_{k,1}, \cdots, X^*_{k,t-1}, X_{k,t})}{P(X^*_{k,1}, \cdots, X^*_{k,t-1}, X_{k,t})}\\
	=&\frac{f_{\theta_{k,1}}(X_{k,t})\Pi_{k,t-1}+(1-\Pi_{k,t-1})pf_{\theta_{k,1}}(X_{k,t})}{f_{\theta_{k,1}}(X_{k,t})\Pi_{k,t-1}+(1-\Pi_{k,t-1})pf_{\theta_{k,1}}(X_{k,t})+(1-\Pi_{k, t-1})(1-p)f_{\theta_{k,0}}(X_{k,t})},
	\end{align*}
	where we use the fact that
	\begin{eqnarray*}
		P(X^*_{k,1}, \cdots, X^*_{k,t-1}, X_{k,t}) &=& 	P(X_{k,t}, \nu_k < t| X^*_{k,1}, \cdots, X^*_{k,t-1})\\
		&&+ P(X_{k,t}, \nu_k = t| X^*_{k,1}, \cdots, X^*_{k,t-1})\\
		&&+P(X_{k,t}, \nu_k > t| X^*_{k,1}, \cdots, X^*_{k,t-1})\\
		&=&f_{\theta_{k,1}}(X_{k,t})\Pi_{k,t-1}+(1-\Pi_{k,t-1})pf_{\theta_{k,1}}(X_{k,t})\\
		&&+(1-\Pi_{k, t-1})(1-p)f_{\theta_{k,0}}(X_{k,t}).
	\end{eqnarray*}
	Let $R_{p,k,t}^{*}$ denote $\frac{\Pi_{k,t}}{ p(1-\Pi_{k,t})}$. Then the recursive formula for $R_{p,k,t}$ is as follows.
	\begin{align*}[left ={R_{p,k,t}= \empheqlbrace}]
	& \frac{f_{\theta_{k,1}}(X_{k,t})}{(1-p)f_{\theta_{k,0}}(X_{k,t})}(R_{p,k,t-1}+1),& \text{if} \quad \delta_{k,t}=1 \\
	& \frac{1}{1-p}(R_{p,k,t-1}+1), &\text{if} \quad \delta_{k,t}=0.
	\end{align*}
	We consider the limiting Bayesian approach by letting $p \to 0,$ and we arrive at the updating rule for  $R_{k,t}^{*} = \lim_{p \to 0} R_{p,k,t}^{*}$ as in  (\ref{eq.R}) for  $R_{k,t}$, except that the initial value $R^*_{k,t=0}$ are from the prior distribution $G$ and the initial value $R_{k,t}$ is $0.$
	
	By using (\ref{eq.R}), a proof by induction shows that $R^*_{k,t}=R_{k,t}+ L_{k,t} R^*_{k,t=0}$ for any time $t=1, 2, \cdots.$ Since $\tilde{R}_{k,t}$  also has the same prior distribution $G$ as $R^*_{k,t=0},$  relation (\ref{eq.random}) holds, completing the proof.
\end{proof}

\subsection{ Proof for Theorem \ref{ARL}}\label{ARL.app}

\begin{proof}
	
	A high-level argument is as follows. In the derivation of the ARL, we observe that $\sum_{k=1}^K R_{k,t}-Kt$ is a martingale under the pre-change hypothesis. Moreover, we have $\sum_{k=1}^r R_{(k),t} \le \sum_{k=1}^K R_{k,t}$. By the optional sampling theorem, for the stopping time $T$ defined in (\ref{stoppingtime}), we show that $\mathbb{E}_\infty[\sum_{k=1}^K R_{k,T}/K]=\mathbb{E}_\infty[T]$, resulting the lower bound. For the upper bound, we define a new stopping time based on any fixed data stream $k$: $T^\prime=\inf\{ t \ge 1: R_{k,t} \ge A\}$. Then we have $\mathbb{E}_\infty[T]\le \mathbb{E}_\infty[T^\prime]=O(A)$, where the last equality follows from Theorem 1 in \cite{pollak1987average}.
	
	Below are the detailed arguments. We first prove the second part of the theorem. For any fixed $k$, $1\le k \le K$, we have $R_{k,t}\le \sum_{k=1}^r R_{(k),t}$. We define a new stopping time based on any fixed data stream $k$: $T^\prime=\inf\{ t \ge 1: R_{k,t} \ge A\}$. Then we have $\mathbb{E}_\infty[T]\le \mathbb{E}_\infty[T^\prime]=O(A)$, where the last equality follows from Theorem 1 in \cite{pollak1987average}.

	To provide a lower bound of $\mathbb{E}_{\infty} T$, we first show that $\sum_{k=1}^K R_{k,t}-Kt$ is a martingale under the pre-change hypothesis as follows.
	\begin{align*}
	&\mathbb{E}_{\infty}[\sum_{k=1}^K R_{k,(t+1)}-K(t+1)|X^*_{k,1}, \cdots, X^*_{k,t-1}, X^*_{k,t}]\\
	=&\mathbb{E}_{\infty}[\sum_{k=1}^K (R_{k,t}+1)\exp(\delta_{k,t+1}\log \frac{f_{\theta_{k,1}}(X_{k,t+1})}{f_{\theta_{k,0}}(X_{k,t+1})})-K(t+1)|X^*_{k,1}, \cdots, X^*_{k,t-1}, X^*_{k,t}]\\
	=&\sum_{k=1}^K (R_{k,t}+1)-K(t+1)\\
	=&\sum_{k=1}^K R_{k,t}-Kt.
	\end{align*}
	Hence, $\mathbb{E}_{\infty} [\sum_{k=1}^K R_{k,T}-KT] = 0$ exists for the stopping time $T$ that are bounded above. Since $\left|\sum_{k=1}^K R_{k,t}\right| < KA/r$ on $\{T > t\}$, we have $\lim \inf_{t \rightarrow \infty} \int_{\{T>t\}}|\sum_{k=1}^K R_{k,t}-Kt|dP_\infty=0$. Therefore, the martingale optional sampling theorem applies to yield $\mathbb{E}_\infty (\sum_{k=1}^K R_{k,T}-KT)=0$. Observing that $\sum_{k=1}^r R_{(k),t} \le \sum_{k=1}^K R_{k,t}$, we have
	\begin{eqnarray}\label{th.eq1}
	\mathbb{E}_\infty[KT]=\mathbb{E}_\infty[\sum_{k=1}^K R_{k,T}]\ge \mathbb{E}_\infty[\sum_{k=1}^r R_{(k),T}]\ge A,
	\end{eqnarray}
	completing the proof.

\end{proof}

\subsection{Proofs for Theorem \ref{prop_incontrol} and Theorem \ref{prop_outcontrol}}\label{incontrol.app}
To prove Theorem \ref{prop_incontrol} and Theorem \ref{prop_outcontrol}, we first present a simple lemma and its proof.

\begin{lemma}\label{lemma2}
	Suppose that there exists a set $U$ such that for all $k \in U$, there exists a $t_0>0$ such that for all $t>t_0$, we have $P(k \in S_{t})=0$. We have $R^*_{k',t}\ge R^*_{k,t}$, for all $ k'\in [K]/U, k\in U$ and for all $ t>t_0$.
\end{lemma}

\begin{proof}
	We will prove by contradiction. Suppose there exists a $t>t_0$,  a $ k \in U$, and a $ k'\in [K]/U$ such that $R^*_{k,t}>R^*_{k',t}$, we consider two cases. If $k'\notin S_t$, then the sensors will first be deployed to the $k$-th data stream then to the $k'$-th data stream because the increments of the two processes are the same if they are unobserved. If $k' \in S_t$, then it indicates that $R^*_{k',t}>R^*_{k,t}$, which leads to contradiction.
\end{proof}

We will use Lemma \ref{lemma2} to prove Theorem \ref{prop_incontrol} and Theorem \ref{prop_outcontrol}.
\begin{proof}[Proof of Theorem \ref{prop_incontrol}]
	Suppose that there exists a set $U$ such that for each $ k \in U$, there exists a $ t>0$, such that for all $ t'>t$, we have $P(k \in S_{t'})=0$.
	
	For any data streams in the set $[K]/U$, we only consider the case when the data is observed, because otherwise, the increments are the same as that of any data streams in the set $U$. For each $ k \in U$, the increment of the logarithmic scale of the local statistic is nearly zero, while the increment for each $k' \in [K]/U$ is nearly $\sum_{t=t_0}^T \log f_{\theta_{k,1}}(X_{k,t})/f_{\theta_{k,0}}(X_{k,t})$. Let $Y_{k,t}=\log f_{\theta_{k,1}}(X_{k,t})/f_{\theta_{k,0}}(X_{k,t})$. It remain to prove $P_{f_{\theta_{k,0}}}\left( \inf_{T<\infty} \sum_{t=t_0}^T Y_{k,t}\ge 0\right) =0$.
	
	Define $Z_n=\sum_{t=t_0}^{t_0+n} Y_{k,t}$ and $\tau_-=\inf\{n: Z_n\le 0 \}$. Notice that $Y_{k,t}$ has negative drift since $\mathbb{E}_{f_{\theta_{k,0}}}\log f_{\theta_{k,1}}(X_{k,t})/f_{\theta_{k,0}}(X_{k,t}) <0$, we have $P_{f_{\theta_{k,0}}}\left( \inf_{n<\infty} Z_n\ge 0\right) = P_{f_{\theta_{k,0}}}\{\tau_-=\infty\}=0$. It means that $R^*_{k,t} $ for any $k\in U$ will eventually get greater than $R^*_{k',t}$ for any $k' \in [K]/U$, which is contradictory to our Lemma \ref{lemma2}.
\end{proof}

\begin{proof}[Proof of Theorem \ref{prop_outcontrol}]
	Suppose we observe the $k$-th data stream at time $t_0$ and it is out of control. Let $Y_{k,t}=\log f_{\theta_{k,1}}(X_{k,t})/f_{\theta_{k,0}}(X_{k,t})$ and $Z_n=\sum_{t=t_0}^{t_0+n} Y_{k,t}$ as in the proof of Theorem \ref{prop_incontrol}. Under the out-of-control state, we have $\mathbb{E}_{f_{\theta_{k,1}}}[Y_{k,t}]>0$.
	By Corollary 8.44 in \cite{siegmund2013sequential}, for $\tau_{-}=\inf\{n: Z_n\le 0 \}$, we have $P(\tau_-=\infty)>0$, which implies $P_{f_{\theta_{k,1}}}\left( \inf_{n<\infty} Z_n\ge 0\right) >0$. It means with positive probability, the local statistic of the $k$-th data stream is always greater than the local statistics of other unobserved data streams. Therefore, we have $P(k\in S_t, \forall t>t_0|k\in S_{t_0})>0$.
	
\end{proof}

\section{Additional Case Study: Solar Flare Data}\label{solar}

We apply the proposed method to a real dataset collected by the Solar Data Observatory, which illustrates an abrupt emergence of a solar flare. The Solar Data Observatory generates 1.5 terabytes of daily data \citep{xie2013change}. Traditional image detection methods that require fully observable data cannot process such high-dimensional and high-velocity datasets on platforms with limited processing power. Besides, when the solar flare incurs, we unsure about its location and the number of affected pixels. Consider the transmission width constraint and the processing power limit; we wish to leverage selected partial data to detect solar flare events. We apply the proposed TSSRP algorithm to this dataset to demonstrate its computational efficiency and monitoring capability.

The data is publicly available at http://nislab.ee.duke.edu/MOUSSE/index.html. It contains a total of $n=300$ frames in the video. Each frame is of size $232\times292$ pixels, which results in $D=67744$ dimensional streaming data. According to the original video, there are two obvious transient flares. One occurred at around $t=187\sim202$, and the other one occurred at $t=216\sim268$. Figure \ref{solarplot} (a and c) are snapshots of the original frame and residual map at $t=210$ before the solar flare emerges, respectively. Figure \ref{solarplot} (b and d) are snapshots at $t=223$ when the second solar flare is brightest. From the figures, we can see that the solar flare signal is sparse. Thus, it is reasonable to monitor this process by sample only a small fraction of the data.

In this study, we first preprocess the data by the MOUSSE algorithm proposed by \cite{xie2013change} to obtain the residual data to remove the background. We verified the normality assumption for the residuals through the Kolmogorov-Smirnov test.  Since there is no change-point before the first 100 solar flares, we treat the first 100 frames as historical observations. Then we standardize the residual data so that the first 100 frames have mean zero and standard deviation one. The threshold is determined by 1000 Monte Carlo simulations of the standard normal distribution, and the IC-ARL is set to be 2500.

We set $r=40, \theta_{k,1}=0.3 $ in the algorithm. We assume that only 2000 out of 67744 pixels are available, i.e., $q=2000$, for data monitoring. Figure \ref{monitorstat} plots the monitoring statistics $\sum_{k=1}^{r}R_{(k),t}^{(0)}$ against $t$, the index of time, based on the top-40 local statistics. The detection threshold is $\exp(10)$. From Figure \ref{monitorstat}, we see that the two solar flares are clear. The monitoring statistic goes sharply up when the solar flare incurs and goes down quickly when it ends. We detect the first solar flare at time $t=192$ and the second one at time $t=217$, comparing to the results of the TRAS algorithm where the two solar flares are detected at time $t=190$ and time $t=221$ in \cite{liu2015adaptive}. Our algorithm is comparable to the algorithm by \cite{xie2013change}, which requires full observations. Their algorithm detects the two solar flares at time $t=191$ and time $t=217$.

\begin{figure}[t]
	\centering
	\includegraphics[width=3in]{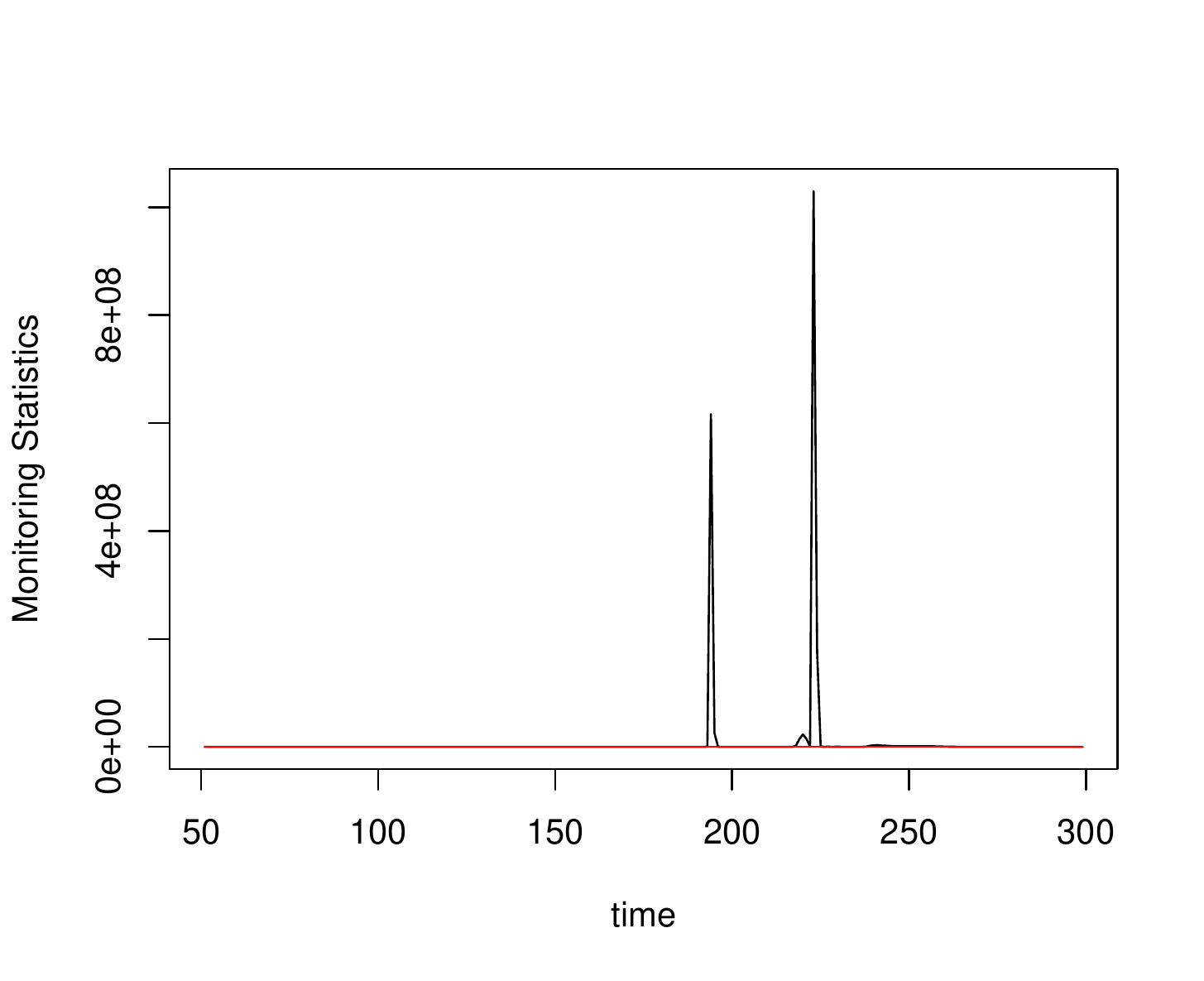}
	\caption{Monitoring statistics against time in the solar flare detection example}
	\label{monitorstat}
\end{figure}

\begin{center}
	\begin{minipage}{.7\linewidth}
		\begin{figure}[H]
			\centering
			\subfloat[][Raw data at $t=210$]{\includegraphics[width=.4\textwidth]{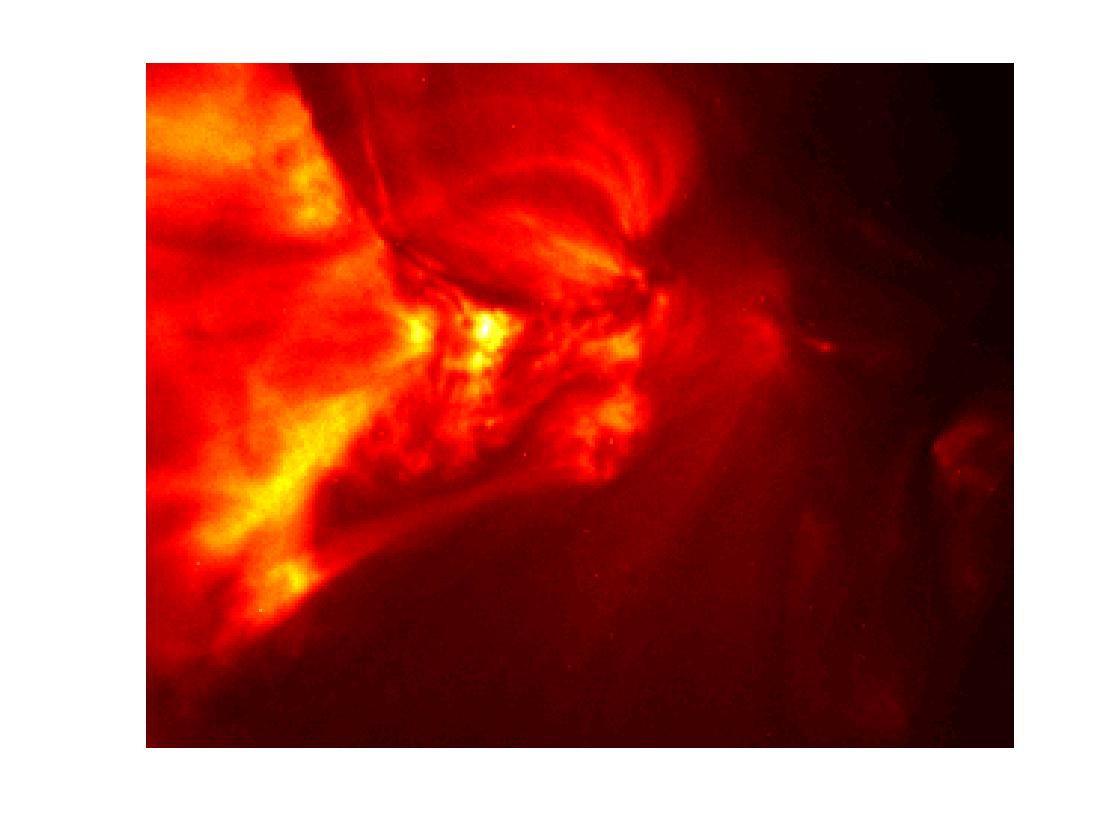}}\quad
			\subfloat[][Raw data at $t=223$]{\includegraphics[width=.4\textwidth]{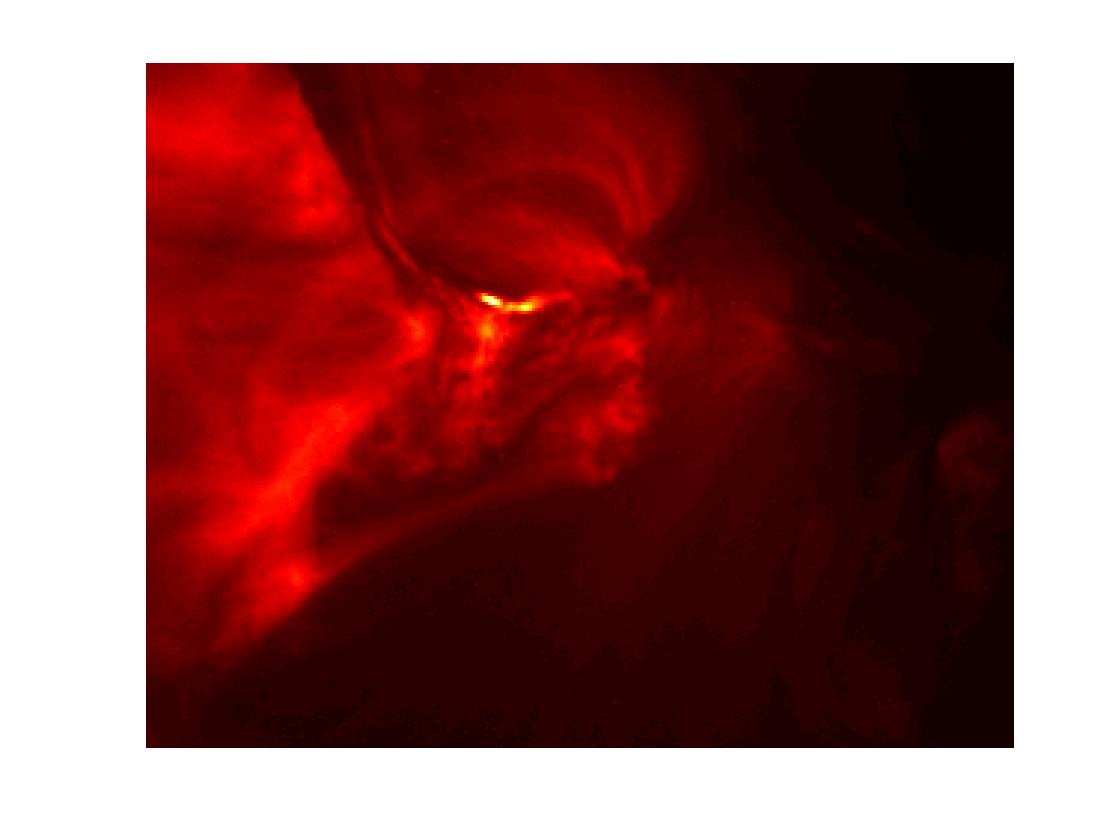}}\\
			\subfloat[][Residual at $t=210$]{\includegraphics[width=.4\textwidth]{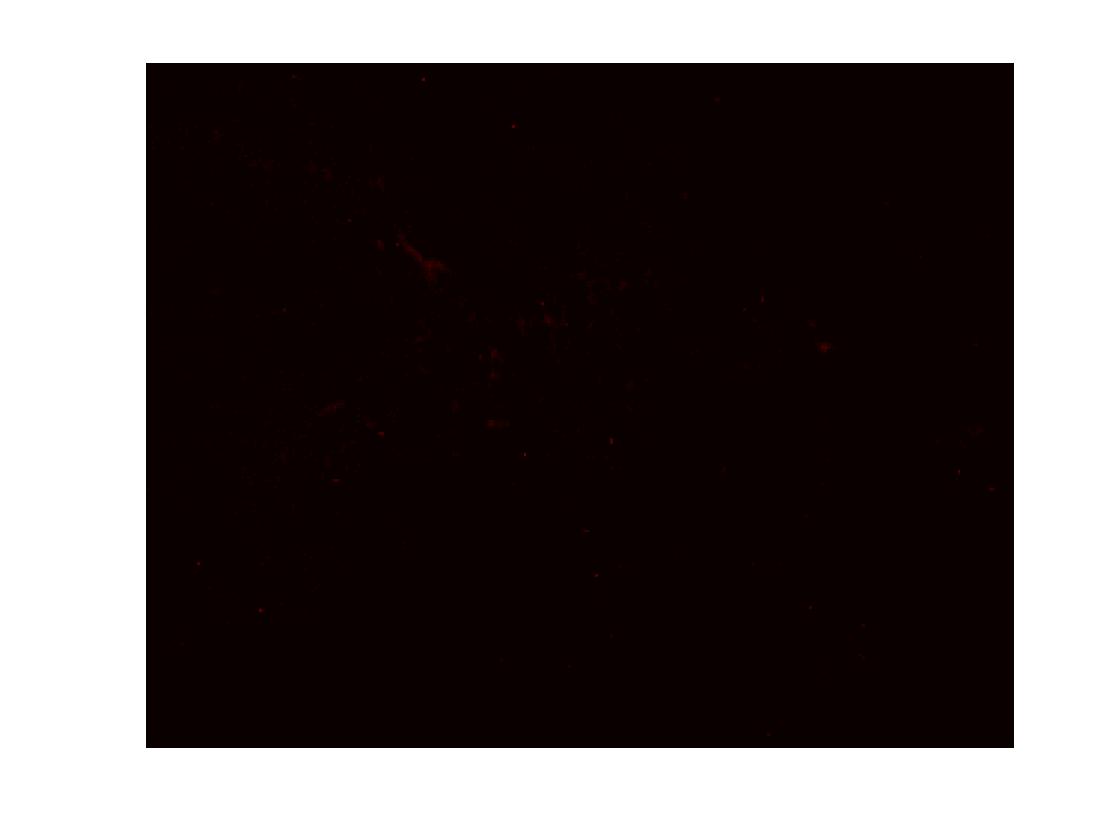}}\quad
			\subfloat[][Residual at $t=223$]{\includegraphics[width=.4\textwidth]{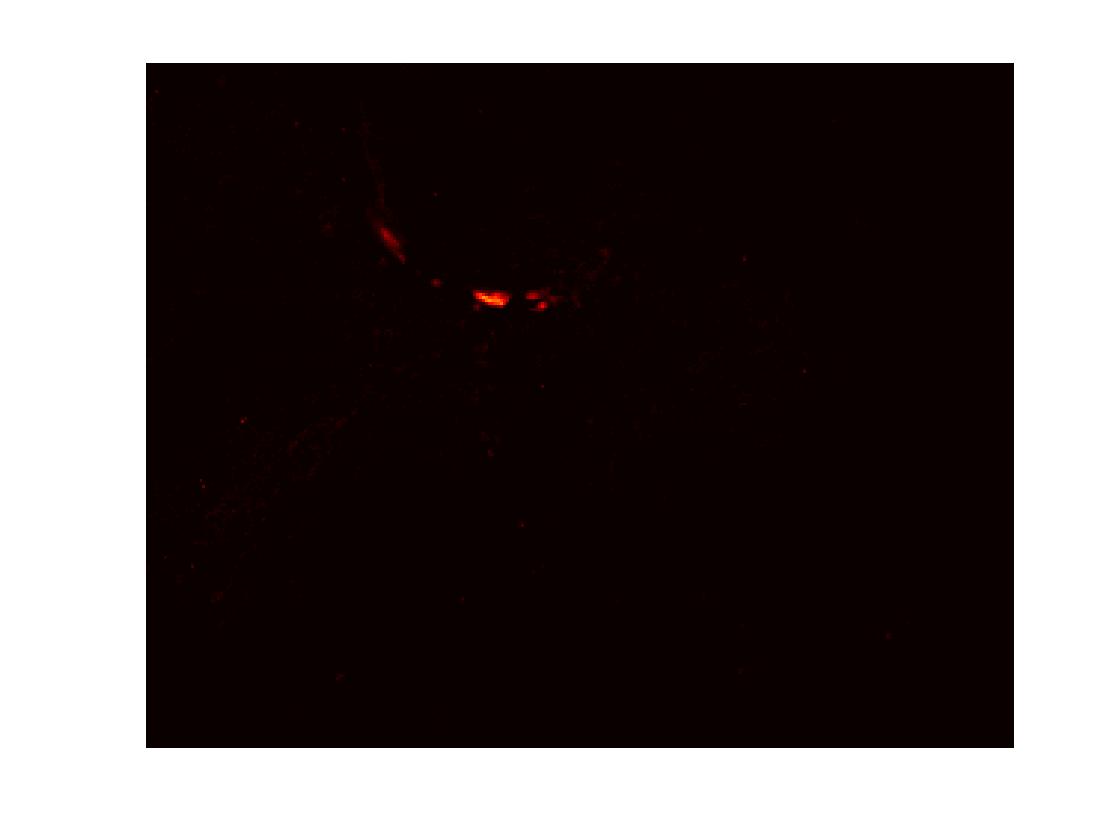}}\\
			\caption{\small Solar flare detection: the snapshot of the video for raw data and residual data at frame $t=210$ when there is no solar flare (a and c); at frame $t=223$ when the solar flare reaches the peak (b and d).
			} \label{solarplot}
		\end{figure}
	\end{minipage}
\end{center}

\section{Additional Simulation Study when $r=3$}\label{app.sim}

Here we present our proposed TSSRP algorithm's properties when $r=3$, i.e., when raising a global alarm based on the sum of the largest $r=3$ local Shiryaev-Robert statistics $R_{k,t}.$

\begin{table*}[h]
	\centering
	\caption{average detection delay under various number of changed data streams for the evaluation of the statistical efficiency experiments when the data is independent mutivariate Gaussian distributed. All the experiments are conducted under $r=3$.}
	\label{table100small}
	\small
	\begin{tabular*}{1\textwidth}{@{\extracolsep{\fill}}lccccccccc@{}}\toprule
The number of changes&1&3&5&8&10\\
\midrule
TSSRP($G_{0}$)&12.537(0.267)&7.668(0.072)&6.761(0.048)&6.106(0.037)&5.814(0.032)\\
TSSRP($G_{1}$)&12.515(0.260)&7.573(0.069)&6.698(0.046)&6.528(0.045)&6.406(0.045)\\
TSSRP($G_{2}$)&18.686(0.332)&11.887(0.146)&9.968(0.103)&8.809(0.084)&8.222(0.076)\\
TSSRP($G_{3}$)&19.729(0.340)&11.884(0.142)&10.038(0.104)&8.654(0.081)&8.128(0.072)\\
TRAS($\delta=0.03$)&32.267(0.589)&18.423(0.265)&13.990(0.187)&11.177(0.137)&9.904(0.114)\\
TRAS($\delta=0.05$)&32.594(0.571)&18.693(0.245)&14.532(0.175)&11.972(0.138)& 10.889(0.119)\\
TRAS($\delta=0.1$)&40.885(0.479)&25.579(0.232)&20.777(0.176)&17.088(0.142)&15.526(0.117)\\

		\bottomrule
	\end{tabular*}
\end{table*}


\section{More global decision policies}\label{stoppingcomp}

Besides the global decision policy defined by the stopping time $T$ in (\ref{stoppingtime}), defined as $T_{1}$ here, we can also consider other global decision policies.

First, we can use the random realizations of $R^*_{k,t}$ for global decision. Let $R^*_{(1),t} \ge \ldots \ge R^*_{(k),t} \ge \ldots \ge R^*_{(m),t}$ denote the decreasing order of randomized local statistics. The stopping time based on the summation of the top-$r$ randomized local statistics is as follows.
\begin{eqnarray}\label{T2}
T_{2}=\inf\{ t \ge 1: \sum_{k=1}^{r}R^*_{(k),t} \ge A\}
\end{eqnarray}

Second, the local statistics can also be the logarithms of  $R_{k,t}$ or $R^*_{k,t}.$ This yields two more decision policies:
\begin{enumerate}
	\item Sum of the logarithm of the top-$r$ classical Shiryaev-Roberts statistics:
	\begin{eqnarray}\label{T3}
	T_{3}=\inf\{ t \ge 1: \sum_{k=1}^{r}\log R_{(k),t}\ge \log A\}
	\end{eqnarray}
	\item Sum of the logarithm of top-$r$ randomized local statistics:
	\begin{eqnarray}\label{T4}
	T_{4}=\inf\{ t \ge 1: \sum_{k=1}^{r}\log R^*_{(k),t} \ge \log A\}
	\end{eqnarray}
\end{enumerate}

We compare the proposed stopping rule with those rules numerically in the simulation experiments under Gaussian distributed data. The simulation is conducted with $r=10$ and $G_2$. We summarize the results in Figure \ref{com.stopping}. It is evident from Figure \ref{com.stopping} that the average detection delay under stopping rule $T$ in (\ref{stoppingtime}) is stable and outperforms other schemes in most cases. The only exception is $T_4$ when the number of true changes gets close to $r=10$. Thus, we pick $T$ as the global stopping rule in the proposed TSSRP algorithm.

\begin{figure}[t]
	\centering
	\includegraphics[width=4in]{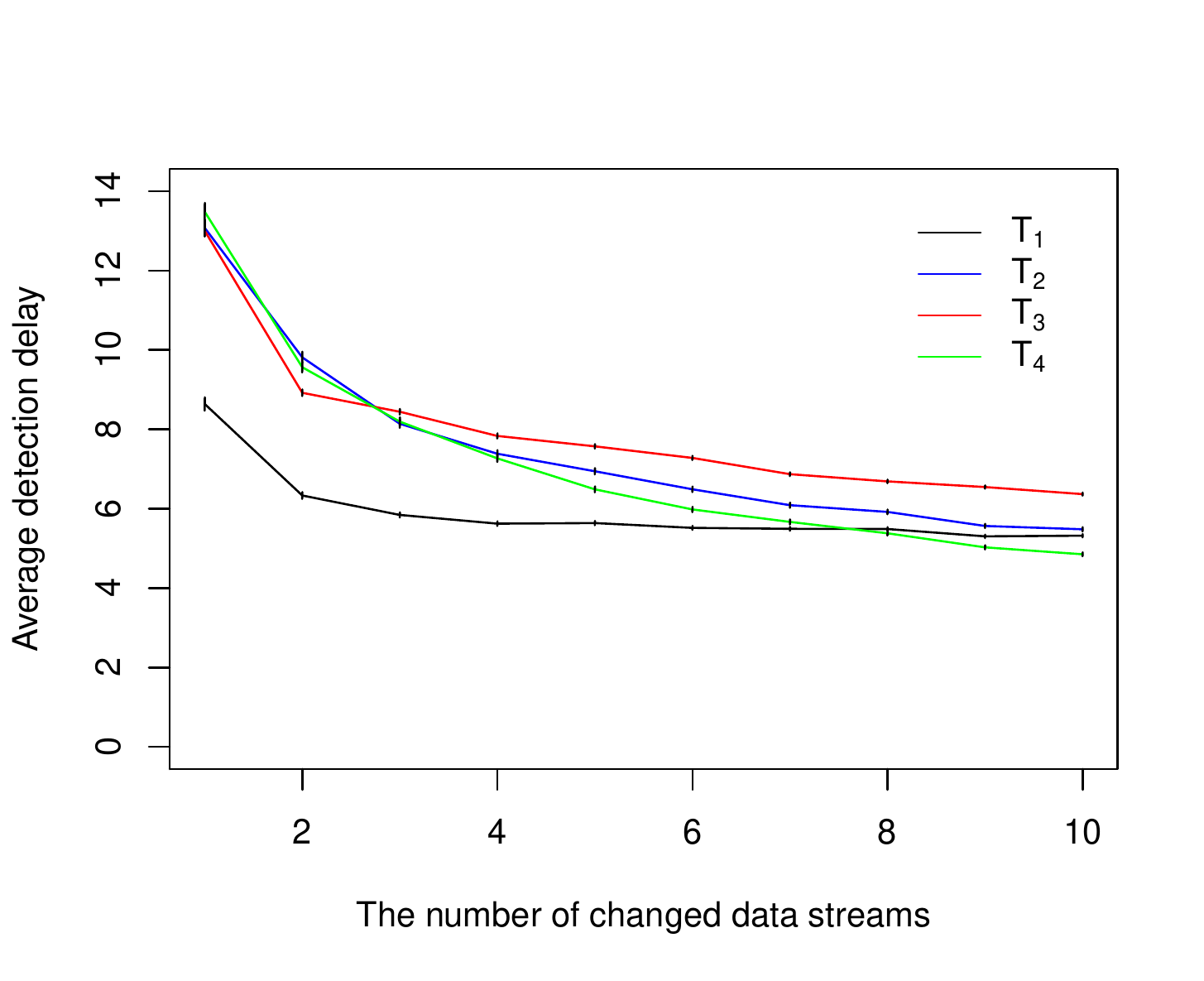}
	\caption{Comparison of four stopping rules:  $T$ in (\ref{stoppingtime}), $T_2$ in (\ref{T2}),  $T_3$ in (\ref{T3}), and  $T_4$ in (\ref{T4}), when the data is independent multivariate Gaussian distributed}
	\label{com.stopping}
\end{figure}

\end{appendices}

\end{document}